\documentclass[english]{llncs}

\usepackage[utf8]{inputenc}
\usepackage{hyperref}

\usepackage[colorinlistoftodos, textwidth=4cm, shadow]{todonotes}

\usepackage{tabularx}
\usepackage{graphicx}
\usepackage{enumerate}
\usepackage{amsmath}
\usepackage{amstext}
\usepackage{amsfonts}

\usepackage{tikz}
\usetikzlibrary[shadows,arrows]
\usepackage{pgf-umlsd}

\newtheorem{thm}{Theorem}

\date{February 28, 2011}

\author{Marek Jawurek, Martin Johns, Florian Kerschbaum\\firstname.lastname@sap.com} \title{Plug-in privacy for Smart Metering billing} 
\institute{SAP Research}
\begin{document} \maketitle

\begin{abstract} 
Traditional electricity meters are replaced by Smart Meters in customers' households.
Smart Meters collects fine-grained utility consumption profiles from customers, which in turn 
enables the introduction of dynamic, time-of-use tariffs.
However, the fine-grained usage data that is compiled in this process also allows to infer the inhabitant's personal schedules and habits. 
We propose a privacy-preserving protocol that enables billing with time-of-use tariffs without disclosing the actual consumption profile to the supplier.
Our approach relies on a zero-knowledge proof based on Pedersen Commitments performed by a plug-in privacy component that is put into the communication link between Smart Meter and supplier's back-end system.
We require no changes to the Smart Meter hardware and only small changes to the software of Smart Meter and back-end system.
In this paper we describe the functional and privacy requirements, the specification and security proof of our solution and give a performance evaluation of a prototypical implementation. 
\end{abstract}


\section{Introduction}
\label{introduction}

\subsection{Motivation}
Smart Meter roll-outs have begun all over the world~\cite{smartmetermap}.
Smart Meters record a fine-grained consumption profile of a certain service (electricity, heat or water) and report it to the supplier of the service who bills the customer accordingly.
Traditionally only a single, compiled value for a whole reporting period  has been reported to the supplier (e.g., the total consumed electrical energy of the last year). In contrast, Smart Meters transmit a detailed set of many data points which document consumption for short time intervals (e.g. every 15 minutes). This enables the suppliers to introduce more dynamic pricing schemes and to collect precise data about their customer base's usage patterns. 

Besides the technical motivation, also legal reasons come into play in respect to the current push towards Smart Metering: 
For instance, in Germany starting October 2010 suppliers must offer either time or load dependent tariffs (see §40~\cite{engw}).
These tariffs necessarily require Smart Meters with fine-grained consumption recording. 

However, such detailed data has privacy implications: 
A listening third party, the supplier or even an employee of the supplier could learn the consumption behavior of a customer and might use this information maliciously for other purposes than intended (e.g, to learn behavioral patterns, such as sleep/wake cycles or vacation time, of a given customer based on his energy usage).
Recently, customers have become aware of the potential privacy implications of such consumption profiles.
In the Netherlands Smart Meter roll-outs have been stopped because of the public outcry about the invasion of customer privacy~\cite{netherlandsstop}.

Grid operators and suppliers now face a dilemma:
On the one hand, they need to implement Smart Metering for legal and technical reasons.
But, on the other hand, they face on-going problems in respect to public acceptance of the technology due to the outlined privacy problems. 

\subsection{Our solution}

We provide a solution to this conflict by introducing a new consumption profile reporting protocol for time-of-use tariffs. 
We introduce a plug-in privacy component into the standardized Smart Meter / Meter Data Management (MDM) reporting communication link.
This component hides the actual consumption profile from the MDM and therefore also from the supplier.
We require only small changes compared to current Smart Meter reporting.
The plug-in privacy component intercepts Smart Meter readings, then uses tariff information provided externally (over the Internet or by the MDM) to calculate the billing amount and sends only the resulting billing amount to the MDM.
A Zero-Knowledge Proof ensures the correctness of the calculation.

The advantages of our approach are the following:
\begin{enumerate}
\item 
The Smart Meter's hardware complexity remains the same, because all calculations are conducted by the stand-alone plug-in component.
Such a plug-in component can be realized by off-the-shelf computing hardware like a router or Wifi access point or even by software running on a standard personal computer.
\item
The supplier does not have to trust the plug-in privacy component.
The privacy component's output suffices to check whether it calculated the final billing amount honestly and correctly, i.e. based on the correct tariff and on the correct readings provided by the Smart Meter. Therefore the privacy component does not require hardware-protected components and can be quite simple and cheap. 
The correct operation of the privacy component can be verified only by its output.

\item
Plaintext, fine-grained consumption profiles never even leave the household, if a privacy component is used.
This prevents any abuse of this data, either by intercepting it in transit, by leakage in the MDM systems or by the MDM's operator himself.
It also spares the MDM expensive security measures for the protection of the massive amount of privacy related data -- the consumption profiles of his customers.
\end{enumerate}

\subsection{Paper outline}

The remainder of this paper is structured as follows:
Section ~\ref{motivation} motivates the problem, gives a short introduction into Smart Metering and its privacy problems and defines our problem statement.
In Section ~\ref{cryptoprimitives} we describe the underlying cryptographic method of our solution before we explain the setup of our solution, the specification of the protocol and its security analysis in Section ~\ref{sec:private_billing_protocol}.
We evaluate a prototypical implementation in Section ~\ref{evaluation}.
Furthermore we show how our protocol might fit into existing Smart Meter communication protocols and how it fulfills the stakeholder requirements.
Finally, we give an overview of related work in Section ~\ref{relatedwork}, provide an outlook on future work in Section ~\ref{futurework} and conclude with a summary in Section ~\ref{conclusion}. 

\section{Smart Metering's implications for Privacy}
\label{motivation}

\subsection{Naming conventions}

Before we explore the Smart Meter billing process and deduct its privacy implications,
we briefly specify the terms used in the rest of this paper:

\begin{description}

 \item [Customer:] The term "customer" represents the household, family or person that receives the service from a supplier.

 \item[Supplier:] The term "supplier" stands as placeholder for all companies that cooperate in order to provide the service to customers and also want to
subsequently invoice the customers for this service. 

 \item[Consumption profile:] The term "consumption profile" stands for the consumption data collected by Smart Meters for service in a certain interval over a certain period of time.
This is applicable to many utilities (electricity, water, heat, gas, etc.). 

 \item[Back-end system:] Usually, the Smart Meter is directly connected to a MDM which just collects consumption data.
Tariffs, are then applied in the supplier's billing systems where the data is subsequently transported to. In this paper, "back-end system" (BS) stands for the collection of all IT-systems that collect consumption profiles and use them to calculate the invoice for the customer based on tariffs.

 \item[Tariff:] The term "tariff" stands for the price schema, i.e., the price of service consumption at a specific interval.
In the following we restrict ourselves to a time-of-use pricing scheme, but our protocol could also handle load-dependent billing with little modification.

\end{description}

\subsection{Smart metering billing}

Smart Metering refers to the collection of consumption profiles at customer's households with the help of so called Smart Meters (SM).
Smart Meters measure electricity consumption in households and communicate their readings at regular intervals to the back-end system.
Alternatively, the  back-end system can also query the Smart Meter for its data (pull).
A Trusted Platform Module (TPM) in the Smart Meter holds key material and creates signatures over the data to ensure authenticity and integrity until it arrives at the back-end system.
There the consumption profile and the tariff data from the respective customer's contract are used to calculated the price the customer has to pay for the time period covered by the profile. 
Figure ~\ref{fig:traditionalsetup} displays the usual Smart Meter setup.

\begin{figure}
\includegraphics[width=\textwidth]{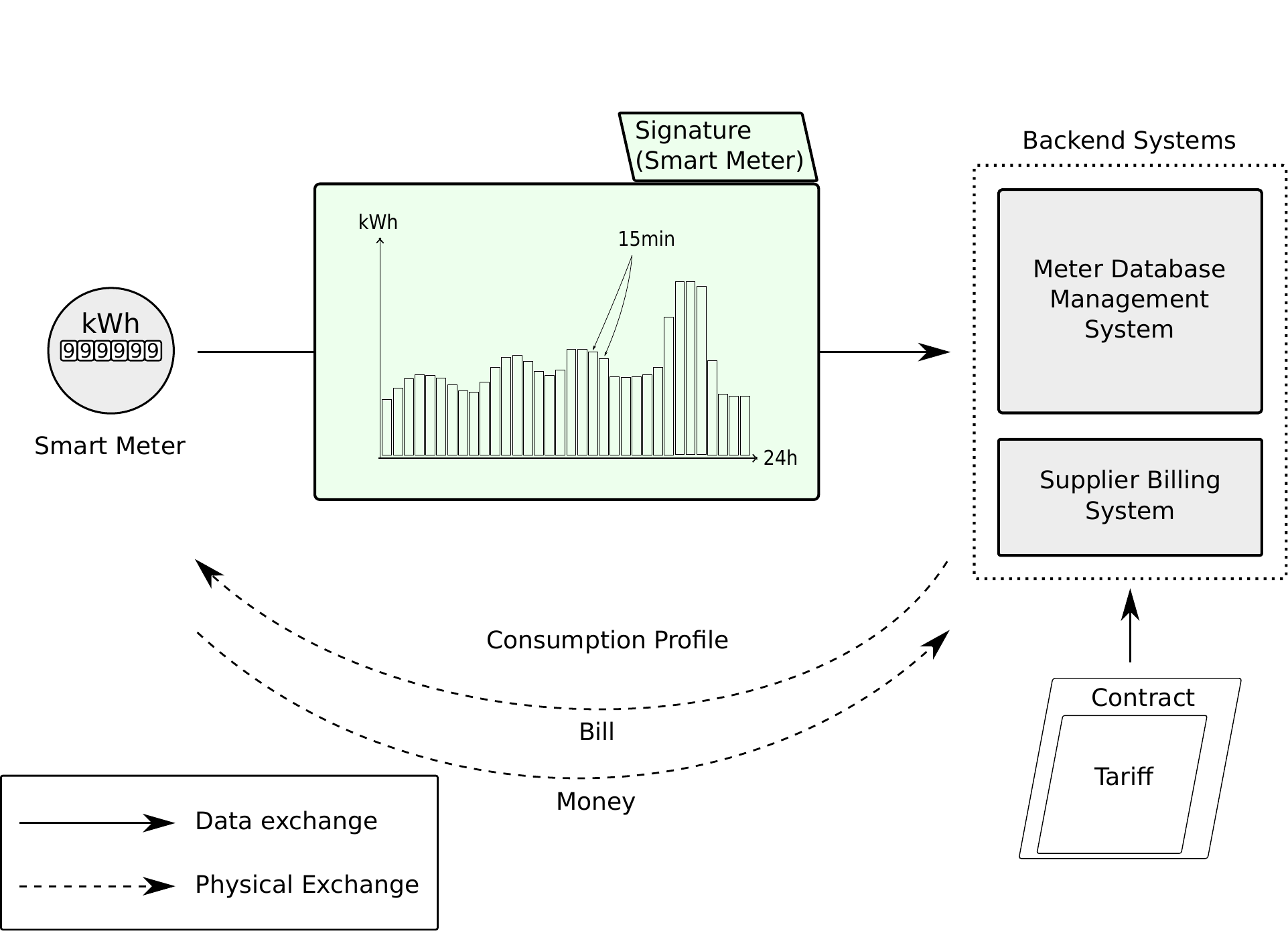}
\caption{Traditional setup of Smart Meter and back-end system}
\label{fig:traditionalsetup}
\end{figure}

\subsection{Privacy concerns}
\label{privacyconcerns}
Smart Metering has encountered massive privacy concerns from media~\cite{ukprivacy}, data privacy experts~\cite{smartgridprivacy} and consumers~\cite{netherlandsstop}. 
The fact that whole consumption profiles of households are transmitted to and stored by suppliers is troubling w.r.t. customer privacy.
Data confidentiality can be easily protected in transit between SM and BS.
However, their storage at the suppliers' IT-systems still endangers customer privacy.
Depending on resolution and the availability of different services' profiles (e.g. water, heat, electricity) one can read the profile and "see" more or less clearly what happens in the household:
For instance, when family members wake up (light switched on), whether they shower in the morning (water, heat, and electricity for water heater), whether they drink hot beverages with their breakfast and when or if they leave for work or school.
Furthermore, the frequency of washing and drying clothes, cooking or the amount of time the TV is turned on can be inferred.
For further research on what electricity consumption profiles tell about the inhabitants see~\cite{appliancemonitoring} or~\cite{usemodekitchen}.

These inferences make consumption profiles very privacy-sensitive data and these profiles might even have value in the advertising market, for instance.
On one hand, disgruntled employees or external attackers might attempt to steal it for profit or out of malice.
On the other hand, the supplier could seek subsidiary revenues by selling this data himself.
Depending on the local jurisdiction, this might even be legal.

The important point is, that currently there are no reliable, technical measures in place to prevent abuse of consumption profiles.
Merely organizational measures, policies or laws sanction the abuse of privacy related data but require a trace or proof of abuse and do not prevent it in the first place.

\subsection{Problem statement}
\label{problemstatement}

The problem we tackle in this paper is to enable suppliers to do billing using Smart Metering data without actually receiving privacy related data.

\subsubsection{Supplier's requirements}
\label{supplier_req}
The supplier's requirement regarding consumption profiles is the ability to reliably use the data in the consumption profile to calculate the customer's bill for received electricity.
The consumption profile $V$ is a vector of $n$ values $v_i$ that represent the amount of utility used in the interval $i$ of one day.
The time-of-use tariff $T$ is a vector of $n$ $t_i$ where interval $i$ is priced with $t_i$. $t_i$ and $v_i$ are integers.  
Then the formula for calculating the time-of-use price for consumption of one day is
$$ P(V,T) = \sum_{i=0}^{n}{t_i * v_i} $$

It is crucial for the supplier that the consumption profiles are accurate and trustworthy.
Clearly, a customer might be inclined to report lower consumption than actually consumed, because it lowers his bill.
Therefore the Smart Meters are equipped with the TPM in order to ensure that the reported consumption profiles are trustworthy and reliable.

\subsubsection{Customers' requirements}
\label{customer_req}
In addition to the requirements of traditional metering (accuracy of the bill), a customer of Smart Metering is concerned about his privacy.
The less information is leaked by the customer, the better for him.
We strive for ideal privacy, i.e. the view of the consumption profile by the supplier is indistinguishable from a uniformly chosen consumption profile with the same price, i.e. supplier obtains no additional information to the price.

\subsubsection{Infrastructure constraints}
\label{infrastructure_req}


A major constraint for the infrastructure investments in Smart Metering is cost.
Suppliers have to replace conventional meters in every household with a new Smart Meter.
This is a significant amount of money for a complete roll-out even for a utilities' provider.
Therefore every technology built into a Smart Meter faces scrutiny w.r.t. to costs.

This also includes the security measures like TPMs and secure storage.
The development and verification of a secure TPM is very expensive and therefore it is common practice to keep its functionality minimal.
One naive approach to privacy-preserving billing of consumption profiles would be to calculate the price in the TPM itself.
But this would require that tariff information are retrieved and verified by the TPM.
In turn, this would require adding an input communication channel and module to the TPM and would consequently increase the costs for building and verifying the TPM considerably.

\subsubsection{Legal constraints}
\label{legal_req}
Depending on the jurisdiction, metering can be subject to legal requirements. In Germany, for instance, metrology laws~\cite{eichg} govern require a certain degree accuracy of the meter and measurements and the tamperproofness of the meter. Privacy laws~\cite{datenschutzg} require the confidentiality of readings to protect consumers' privacy.
We translate this into the technical requirements of Smart Meter integrity and integrity and confidentiality of consumption profiles on the wire and in computer systems.

\section{Pedersen Commitments}
\label{cryptoprimitives}

The core of our proposed solution (which we present in Sec.~\ref{sec:private_billing_protocol}) relies on Pedersen Commitments~\cite{Ped}.
In this section we briefly introduce the basics of this cryptographic method.
For further information on the scheme please refer to~\cite{Ped}.

A commitment is a cryptographic tool with two functions:

\begin{itemize}

\item {\em Commit($x$, $r$)} $\longrightarrow c$ takes as input a value $x$ and a random number $r$.
As output it produces the commitment $c$.

\item {\em Open($c$, $x$, $r$)} $\longrightarrow \top/\bot$ takes as input a commitment $c$, a value $x$ and a random number $r$.
It outputs $\top$, if $c$ is indeed a commitment to $x$ and $\bot$, if not.

\end{itemize}

Commitments have two security properties:

\begin{itemize}

\item {\em Secret}: Given $c$ it is hard to compute $x$.

\item {\em Binding}: Given $c$, $x$ and $r$ it is hard to compute an $x' \neq x$ and $r'$, such that $Open(c, x', r') = \top$, i.e. $c$ is a commitment for $x'$ as well.

\end{itemize}

They are used in the following way:
Alice chooses a value $x$.
She computes a commitment $c$ and sends it to Bob.
Now, Alice and Bob may, for example, engage in some computation that depends on Alice's input $x$, but where Alice may no longer change her mind.
Alice opens her commitment and shows that everything was indeed computed according to the value $x$ she choose at the beginning.

A typical example is fair coin flip.
Alice chooses a random $s$ and sends the commitment $c$ of $s$ to Bob.
Bob chooses a random number $t$ and sends it to Alice.
Alice now opens her commitment.
The fair coin flip is $x = s \oplus t$ (where $\oplus$ denotes ``exclusive-or'').
If the commitment was not secret, Bob could choose $x$.
If the commitment was not binding, Alice could choose $x$.

Pedersen commitments operate over a group $\mathbb{G}$.
This group $\mathbb{G}$ can be the same elliptic curves as used of EC-DSA in the secure hardware of the Smart Meter.
Let $g$ and $h$ be two generators of $\mathbb{G}$.
Pedersen commitments are computed as follows:

\newcommand{\entspricht}{\mathrel{\widehat{=}}}
\newcommand{\isequal}{\stackrel{?}{=}}

\begin{itemize}

\item {\em Commit($x$, $r$)}:
\[
c = g^x h^r
\]

\item {\em Open($c$, $x$, $r$)}:
\[
c \stackrel{?}{=} g^x h^r 
\entspricht c \isequal Commit(x,r)
\]

\end{itemize}

The proofs of their security properties can be found in~\cite{Ped}.

Pedersen commitments have another very useful property we exploit in this paper.
They are homomorphic, i.e. a multiplication of two commitments results in a commitment to the sum of their committed values.
\[
Commit(x, r) Commit(y, s) = Commit(x + y, r + s)
\]

A commitment can also be multiplied by a plain factor $y$
\[
Commit(x, r)^y = Commit(x y, r y)
\]

Note that both operations change the commitment $c$, such that the binding security property is not violated.
Instead one needs to open with the new input values of the commitment.

\section{The private billing protocol}
\label{sec:private_billing_protocol}
In this section we describe our privacy-preserving Smart Meter billing protocol.
First we give a very abstract description in Section ~\ref{abstractsolution}, then we provide the full specification in Section ~\ref{concretesolution} and provide a security analysis in Section ~\ref{analysis}.

\subsection{Components and specification}
\label{abstractsolution}
The main idea of our approach is that the plaintext consumption profiles never leave the household, but only after they have been processed by a pseudo-random one-way function.
Therefore, ideal privacy is preserved.
We propose to introduce a privacy component (PC) into the communication link of the Smart Meter and the supplier's back-end system.
Its objective is to intercept reports of consumption profiles and to let only processed information pass-through.
The PC is invisible to the SM and only the supplier will notice it: The PC directly interacts with the supplier's systems and consumption profiles will look different if a PC is used.
This setup is illustrated in Figure ~\ref{fig:proposedsolution}. 

\begin{figure}
\includegraphics[width=\textwidth]{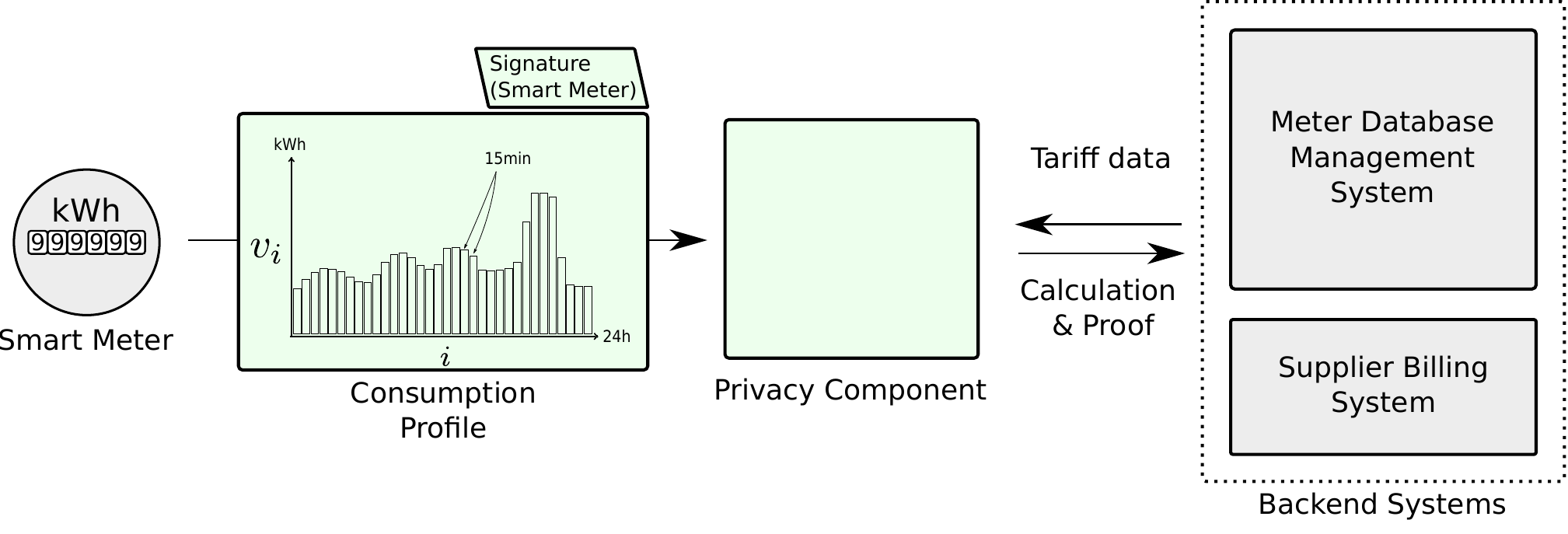}
\caption{Setup of proposed solution with intermediate privacy component}
\label{fig:proposedsolution}
\end{figure}

The major difference to a standard Smart Metering setup is that the price function $P(V,T)$ is not calculated at the supplier's system.
It is calculated in the PC which is supposed to be located in the household. 
For this, the PC intercepts the consumption profile and signed commitments sent to it by the Smart Meter 
and removes the plaintext consumption profile.
Then the PC obtains the tariff information from the supplier and calculates the bill with the original consumption profile.
It then presents the invoice, the signed commitments and a Zero Knowledge Proof to the supplier who verifies the bill's validity using the homomorphic property of the used commitment scheme:
The supplier determines the correctness of the bill by appropriate operations on the received signed commitments and the tariff.
If the commitments can be verified on the presented bill, then the presented bill is trustworthy and correct.
The homomorphic commitment scheme we use on the Smart Meter side is Pedersen Commitment~\cite{Ped} and is shortly outlined in Section ~\ref{cryptoprimitives}.

\subsection{Protocol specification}
\label{concretesolution}

\subsubsection{Initiation}
\label{initiation}
Initially, the SM and BS need to employ some signature scheme which allows the SM to secure the integrity of data sent to the BS. This is usually already the case with Smart Meters. They are either part of a PKI or both, the SM and the BS, have access to a symmetric key for a symmetric signing scheme. We denote such signing key with $\text{Sign}_\text{priv}$.

Secondly, the TPM in the Smart Meter must be able to use the Pedersen Commitment scheme (see Section ~\ref{cryptoprimitives}) over elliptic curves with public generators $g$ and $h$.

How keys (or the public parameters, such as the generators) are distributed to Smart Meters is beyond the scope of this paper, but it is already common practice in on-going Smart Meter deployments.

\subsubsection{Consumption profile reporting and invoice calculation}
\label{calculation}

\begin{figure}
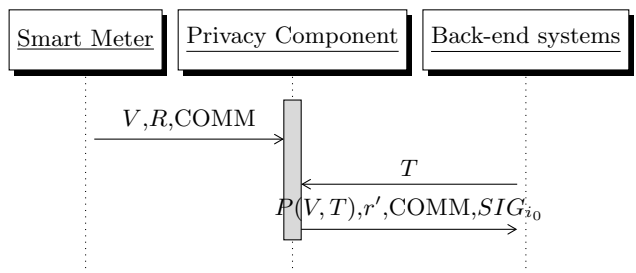

\centering

 \begin{sequencediagram}

    \newinst{sm}{Smart Meter}
    \newthread{pc}{Privacy Component}
    \newinst{bs}{Back-end systems}

  \mess{sm}{$V$,$R$,$\text{COMM}$}{pc}

  \mess{bs}{$T$}{pc}
  \mess{pc}{$P(V,T)$,$r'$,$\text{COMM}$,$SIG_{i_0}$}{bs}

  \end{sequencediagram}
\caption{Communication sequence}
\label{fig:sequence}
\end{figure}

Figure ~\ref{fig:sequence} illustrates the communication that takes place between the different actors and the following enumeration of steps describes the protocol in detail:
\begin{enumerate}
\item The SM prepares a consumption profile to be reported to BS. The profile basically consists of a vector of consumption values  
$V=\{v_{i_0},v_{i_1},...,v_{i_n}\}$. $v_{i_k}$ represents the energy consumption of the household in the interval $i_k$. $i_k$ stands for the interval number, 
incremented since a fictive first interval, analogous to the definition of the UNIX time stamp. 
\item For values in $V$ the SM now creates commitments. The commitment of $v_{i_k}$ is $Comm_{i_k}=Commit(v_{i_k},r_{i_k}) $. Where $Commit(a,r)$ stands for the Pedersen Commitment  of $a$ and a random value $r$ with the generators $g$ and $h$ known to the TPM and the BS.
\item \label{sigi} Now, the SM would like to send the data to the BS. Before that can happen, it protects the data from being manipulated on the way. It creates a signature $SIG_{i_0}$ (with its signing key $\text{Sign}_\text{priv}$) over $(i_0,\text{COMM})$ and sends it together with the vector $V$, the vector $\text{COMM} = \{Comm{i_0},Comm_{i_1},...,Comm_{i_n}\}$ and the vector $R = \{r_{i_0},r_{i_1},...,r_{i_n}\}$ towards the supplier's back-end system.
\item \label{intercept} The PC intercepts all the traffic between the SM and the BS.
\item \label{tariff} The PC obtains the tariff vector $T=\{t_{i_0},t_{i_1},...t_{i_n}\}$ from BS and performs the following two calculations:
\begin{enumerate}
\item $P(V,T)=\sum_{k=i_0}^{i_n} v_k*t_k $  This is the actual price the customer has to pay 
for the reporting period represented by $V$. 
\item In addition it also calculates $r'$ from the vector $R$ it intercepted in step ~\ref{intercept}:
$r\prime=\sum_{k=i_0}^{i_n} r_k * t_k $ 
\end{enumerate}
\item \label{proof} The PC now sends $P(V,T)$, $r'$, $\text{COMM}$, $SIG_{i_0}$ to BS and has finished its work.
\end{enumerate}

\subsubsection{Verification}
\label{verification}
These are the steps  performed by the BS subsequently  to the reporting in order to verify that the $P(V,T)$ was correctly calculated:

First of all, the BS verifies that the signature $SIG_{i_0}$ over $i_0$ and the commitments is intact which means that the commitments it received has been signed by the TPM and stands for the next vector $V=\{v_{i_0},v_{i_1},...,v_{i_n}\}$ starting from $i_0$.

BS now computes $\text{COMM}_{\text{Tariff}}$ with the $Comm_i$ it received in step ~\ref{proof} and the tariff vector $T$ that it made available to PC in step ~\ref{tariff}: 
\begin{equation*}
\text{COMM}_{\text{Tariff}}=\prod_{k=i_0}^{i_n} Comm_{k}^{t_k}
\end{equation*}
Whether the $P(V,T)$ sent by the PC has been calculated truthfully with the correct $v_i$ and $t_i$ can now be verified by opening the aggregated commitment $\text{COMM}_{\text{Tariff}}$. For that, the BS uses $P(V,T)$ and the aggregate random value $r'$ that it received in step ~\ref{proof}.
\begin{equation*}
\begin{split}
Open(\text{COMM}_{\text{Tariff}},P(V,T), r\prime)\\
= \text{COMM}_{\text{Tariff}} \isequal Commit(P(V,T),r\prime) 
\end{split}
\end{equation*}

\subsection{Analysis}
\label{analysis}
\begin{thm}
Our private billing protocol is complete, sound and honest-verifier zero-knowledge.
\end{thm}

\begin{proof}
For completeness, i.e. if the PC truthfully computes the tariff, then the BS accepts, we observe the following equation:

\begin{equation*}
\begin{split}
  & Commit(P(V,T), r\prime) \\
= & Commit(\sum_{k=i_0}^{i_n} t_k v_k, \sum_{k=i_0}^{i_n} t_k r_k)\\
= & \prod_{k=i_0}^{i_n} Commit(t_k v_k, t_k r_k) \\
= & \prod_{k=i_0}^{i_n} Commit(v_k, r_k)^{t_k} \\
= & \prod_{k=i_0}^{i_n} Comm_k^{t_k} \\
= & \text{COMM}_{\text{Tariff}}
\end{split}
\end{equation*}

It follows that $\text{COMM}_{\text{Tariff}}$ is a Pedersen commitment for $P(V,T)$ with the random number $r'$.

For soundness we prove that if the PC does not truthfully compute the tariff, then the BS must reject.
That is given $v_i$, the PC cannot forge a view $Comm_i$, $P(V, T)$ and $r'$ of the protocol that is accepted by the BS.

We will prove by contradiction.
First, observe that we assume that the PC cannot forge the $Comm_i$ commitments, since they are signed by the TPM.
Second, as follows from completeness, the subsequently computed $\text{COMM}_{\text{Tariff}}$ is a Pedersen commitment to $P(V,T)$ and $r'$.
If, the PC could present $P'(V,T) \neq P(V,T)$ and $r''$, such that $\text{COMM}_{\text{Tariff}} = Commit(P'(V,T), r'')$ is opened correctly, then this would be a contradiction to the binding property of Pedersen commitments as established in \cite{Ped}.

For honest-verifier zero-knowledge, we present a simulator of the view of the BS given only its input and output.
The values $Comm_i$ and $r'$ from the view of the protocol are uniformly and independently distributed in $\mathbb{Z}^*_q$.
The tariff $P(V, T)$ is public output of the protocol (and input to the verification operation).

The signature $Sig_{i_0}(Comm_1, \ldots, Comm_n)$ of the TPM cannot be trivially simulated, since the BS only holds the public key.
Nevertheless, since it is only a signature of randomly distributed values, we could simulate it by inverting the signature verification operation on a random signature.
This rather strange simulation is an artifact of our unconventional setup of proving having the PC compute on input from another party -- the TPM.
In a strict sense, the signature is not part of the Zero Knowledge Proof, since it is computed by the TPM and not the PC.

\end{proof}

\section{Implementation and Evaluation}
\label{evaluation}

In this section, we give details on our prototypical implementation, show how our component can be integrated in real world Smart Meter deployments, and discuss how our solution fulfills the functional and security requirements which were identifies in Section~\ref{problemstatement}. 

\subsection{Implementation of the core algorithm}\label{implmentation}

We implemented an exemplary system to identify load on the respective hardware systems during the execution of our protocol. 
For this purpose, we modeled the SM, the PC and the 
BS in Java as much as necessary to execute our protocol. The SM creates a profile, creates respective commitments and signs commitments. Transport of data is simplified by method calls to the respective destination component. Obtained data shows the load on the CPU  in the different components. Figure ~\ref{table:cpuprotocol} shows execution times of the different components (SM,PC,BS) of the reporting and verification part of the protocol. Execution times were measured on a Intel(R) Core(TM) i5 CPU M540 at 2.53GHz on a OpenJDK Runtime Environment (IcedTea6 1.8.1) on a Ubuntu 10.4 system.


\begin{figure}
\includegraphics[width=\textwidth]{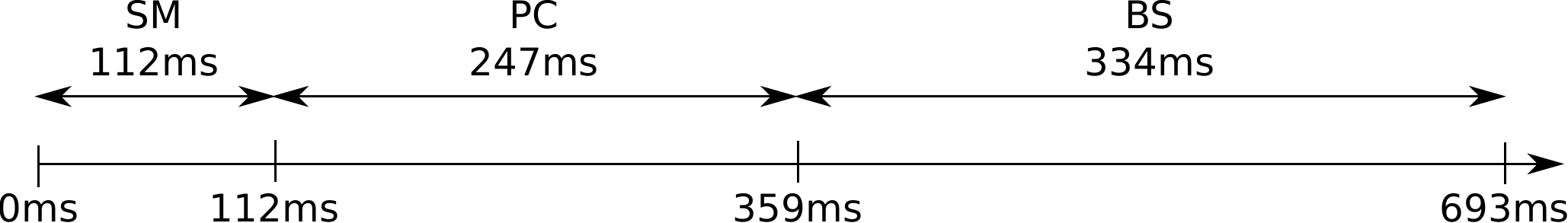}
\caption{Execution times of the protocol (reporting and verification stages) at the different entities}
\label{table:cpuprotocol}
\end{figure}

From the numbers in the Figure ~\ref{table:cpuprotocol} one can see that most time is spent in the BS and the PC. The SM does not spend excessively much time with creating commitment values. 

Although the hardware of the Smart Meter, respectively the TPM performing the actual calculations, is usually several scales inferior to our test system we believe that the SM is able perform its part of the protocol in a timely manner \cite{standardjavacard}.
After all, irrespective of other constraints, it has one day before it needs to perform the next protocol run.
The PC can be realized either by a stand-alone hardware component or by software running on existing hardware (router, Wifi access points). Therefore, the hardware can be chosen appropriately to stand up to the requirements of its part of the protocol.

However, one has to take into account that the supplier's systems will need to participate in several thousand instances of this protocol per day, one for every associated household. If we assume that the supplier buffers the received verification data of concurrent protocol instances it can spread verification (where all of its time is spent) over the course of a whole day. Then, one such system (with our hardware) should be able to handle approx. 25000 protocol instances per day. This could be further increased by only verifying a random choice of reported prices. This shows, that a supplier should be able to implement our protocol for millions of users with negligible resources.

\subsection{Integration in real world scenario}\label{realworld}
Regarding the integration of our approach into current Smart Meter setups we have identified how our approach could use existing protocols of Smart Meter reporting. 
From \cite{smartmeterstandardcomparison} we identified two relevant application layer protocol specifications for SM to BS reporting: The universal DLMS/COSEM standard suite (IEC 62056 / EN 13757-1)~\cite{dlmscosem} and the simple Smart Metering Language~ (SML)~\cite{sml} specification. 

The Smart Meter Language (SML) describes an application and presentation layer and Smart Meters operate either in a push (SM initiates) or pull (SM reacts) scenario. 
All data is encoded in either an SML request or SML response message. Encryption of SML messages on the application or presentation layer is not part of the SML specification.

DLMS/COSEM is an application layer protocol. DLMS specifies how one can talk about energy metering objects. Energy Metering objects are described by the COSEM specification.The standard does not dictate specific transport protocols. The Smart Meter operates as server and communication follows the pull-strategy from the view of the BS system. Read and write access is realized by transmitting respective COSEM objects in APDUs (Application Protocol Data Units). The server's application context determines whether APDUs are encrypted.

How the privacy component can be embedded in environments employing SML or DLMS/COSEM depends on a multitude of factors: The actual protocols used on the network/transport layers, the used push/pull strategy as well whether encryption is used. Those factors determine whether the PC acts as transparent or visible proxy, how it intercepts messages and whether it needs key material to decipher messages.

For SML we see a simple straightforward solution how to implement the privacy component on the application layer:
In SML actual consumption profiles are sent as table with one row for each recorded value. The columns can describe one entry further with entries like time of recording, error conditions and so on. The whole table but also individual columns of the table can be signed which would allow us to fit our protocol into SML messages easily: For every query (pull-scenario) of consumption values the SM would answer with a table with the columns: $i$,$v_i$,$Comm_i$,$r_i$. The SM would sign all columns independently but the PC would intercept the SML response and delete the columns $v_i$ and $r_i$ from the table and insert $P(V,T)$,$r'$ and $\text{COMM}$ into the message. $Sig_{i_0}$ would be represented by the $i$'s and $Comm_i$'s columns' intact signatures. This only requires, that the part of the Smart Meter responsible for creating signatures also create commitments. The BS system would notice that columns $v_i$ and $r_i$ are missing and would therefore switch into a mode where it communicates with a privacy component and performs the verification part of our protocol. If a privacy component was not employed the whole table would be transferred intact and the BS system would perform its normal operation and store the plaintext values in its database.
 
For DMLS/COSEM the approach would work analogously but with COSEM objects and properties instead of SML tables. However, in DLMS/COSEM encryption could make it impossible for the PC to understand the intercepted APDUs. In such a case, the Smart Meter would either need to be reconfigured not to use encryption or to use the public key of the PC instead of the supplier's public key.
This would allow the PC to read and manipulate the APDU and possibly re-encrypt it for the supplier with the supplier's public key.

\subsection{Fulfillment of stakeholder's requirements}
In Section~\ref{problemstatement} we listed requirements of the different stakeholders for Smart Metering.
We will show in this Section how our approach fulfills these requirements.

\begin{itemize}

\item In Section ~\ref{supplier_req} we mentioned that the supplier's requirement is the trustworthiness of reported consumption values.
Our protocol fulfills this by providing a trustworthy price instead of individual consumption values. 
We have given a soundness proof of our Zero Knowledge proof for the correct calculation of the price.

\item In Section ~\ref{customer_req} we also stated that the customer wants privacy-aware billing.
Our approach achieves a privacy-aware billing with ideal privacy as the consumption profile never leaves the household unprocessed, but only the price.
We have proven the zero knowledge property of our Zero Knowledge Proof, i.e. the supplier will learn nothing, but the price.

\item The infrastructure requirement listed in Section ~\ref{infrastructure_req} is a low-cost Smart Metering solution.
Our approach achieves this by only minimal changes to the software of Smart Meters and supplier's back-end system.
The privacy component itself is simple and untrustworthy.
It could therefore be implemented in inexpensive hardware or even in software.

\item  Finally, in Section ~\ref{legal_req} we mention several requirements regarding the tamperproofness, accuracy and confidentiality of the Smart Meter. As our approach does not interfere with the Smart Meter's normal operation accuracy and tamperproofness of the Smart Meter are not changed.
We conform to any regulation we are aware of.

In addition, the supplier might benefit from the use of a privacy component as well, as less privacy-related data has to be stored in his systems for legal retention periods. The supplier needs to store all data that he receives from the PC for being able to reproduce invoice calculation for a certain retention time but that data is not privacy-related. The commitment values do not disclose useful information and the only privacy-related data item is the final price. This reduces the supplier's need for special security measures of his systems against internal or external attackers.
\end{itemize}
Based on the discussion above, it is save to
conclude that our approach fulfills all identified requirements (see Section ~\ref{problemstatement}) for a privacy-respecting billing of Smart Metering consumption profiles. Furthermore, as shown in Sections~\ref{implmentation} and~\ref{realworld}, an implementation of our algorithm is suitable to be deployed on a large scale and fits well with existing standards and infrastructures. 

\section{Related work}
\label{relatedwork}

\paragraph{General references concerning security aspects of smart metering:} 
Abstract predictions about security and privacy challenges potentially coming with the evolution of the current grid to the Smart Grid are described in \cite{johnsjawurek}
while \cite{trustsecprivforsmartgrid},\cite{secprivchallengessg} and \cite{smartgridprivacy} give more information on the topics of security, privacy and trust in Smart Grids/Smart Metering and Advanced Metering Infrastructures (AMI).

In \cite{smartgridprivacy} Ontario's (Canada) Information and Privacy Commissioner provides an overview of what the Smart Grid is, how it will affect electricity consumers and how their privacy might be at risk by the Smart Grid and Smart Metering. Furthermore she promotes the idea of building privacy into the Smart Grid from the start.

\paragraph{Privacy aspects of smart meter-based billing:}
In \cite{GarJac} a privacy-preserving detection algorithm for leakages in electricity distribution has been proposed.
By aggregation across several Smart Meters the developed algorithm protects individual meter readings while allowing grid operators to detect illegitimate/unknown load.
Their approach does not allow individual billing, yet this is the main application of our paper.

Furthermore, in \cite{privacymodel} a model for measuring privacy in Smart Metering is developed and subsequently two different solutions to privacy are presented: A Trusted Third Party-based approach, where individual consumption profiles are aggregated at the third party and only sums are communicated to the supplier. The other approach attempts to mask consumption profiles by adding randomness to the actual profile with an expectation of the random distribution of zero. In contrast to our solution, both of their approaches cannot handle billing of time-of-use tariffs but only provide either sums or not-accurate profiles. Furthermore, our approach does not require a trusted third party and provides exact results for every computation (as required by some legislations). 

Finally, in~\cite{privconcept} also a twofold approach is presented: The first solution employs a sophisticated Trusted Platform Module (TPM) in the Smart Meter to obtain signed tariff data from the supplier and calculate a trustworthy bill. The second solution makes use of the electrical grid infrastructure as a third party to anonymize up-to-date consumption values sent out constantly by Smart Meters. Our approach can be distinguished as it only addresses billing but only requires a very simple TPM that creates commitments.

\paragraph{Pedersen commitments:}
Due to their homomorphic properties Pedersen commitments are an effective means to verify the correctness of statistics computation.
In~\cite{ThoYao} it has been applied in the outsourced database setting.
Statistics and dot product computation can be useful for in many application areas.
An example from the database community again is privacy-preserving data mining~\cite{VaJZhu}.
An example from the business software community is collaborative benchmarking~\cite{Ker}.

Our work is the first in providing a very high degree of privacy for customers by not disclosing consumption profiles in time-of-use Smart Meter billing scenarios.

\section{Future work}
\label{futurework}
Dynamic, time-depended billing is only one application of fine-grained consumption data. In addition, 
a profile of a household's energy consumption can be utilized by the supplier to create predictions of this household's energy demand in the future. Our proposed solution does not cover this usage of consumption profiles. 
Realizing a privacy friendly method for calculating such predictions is subject to future research. 
However, one must realize that privacy and the ability to create predictions potentially contradict each other in their nature and this conflict in the field of Smart Metering should also be investigated further.


\section{Conclusion}
\label{conclusion}
In this paper we have proposed a protocol for privacy-preserving reporting of consumption profiles in a Smart Metering scenario by the use of a plug-in component. We have identified and analyzed the requirements of different stakeholders. Based on this analysis, we devised a billing scheme  which allows privacy-related consumption profiles to remain within the household while preserving provable correctness of the billable amounts. 
The privacy sensitive data therefore is not susceptible to interception in transit or leakage in the supplier's back-end system. 

We have provided the specification for the utilized components, for the introduction into a traditional Smart Metering setup, and for the communication and calculation during the three protocol stages (initialization, reporting, verification). After proving the soundness, completeness and zero-knowledge property of the verification, we investigated the execution times of our prototypical implementation and showed that it is a viable solution for Smart Metering hardware. 
Finally, we discussed how our protocol could be executed using existing Smart Metering reporting specifications and showed that our approach fulfills the previously identified stakeholder's requirements.

Our protocol is one step towards the idea of building privacy into the Smart Grid~\cite{smartgridprivacy}. By preserving customer privacy we mitigate trust issues that privacy experts, the media and the public have raised about the privacy implications of Smart Metering.

\bibliographystyle{abbrv}
\bibliography{privatebilling}

\end{document}